\documentclass[letterpaper]{article}
\usepackage{aaai}
\usepackage{times}
\usepackage{helvet}
\usepackage{courier}
\usepackage{hyperref}
\usepackage{algorithm}
\usepackage{algpseudocode}
\usepackage{amsmath}
\usepackage{amsthm}
\usepackage{amsfonts}
\newtheorem{theorem}{Theorem}
\newtheorem{lemma}{Lemma}
\newtheorem{definition}{Definition}
\theoremstyle{definition}
\newtheorem{example}{Example}
\begin{document}
\title{At Least Factor-of-Two Optimization for RWLE-Based Homomorphic Encryption}
\author{Jonathan Ly \\ University of Washington \\ jly02@uw.edu
}
\maketitle
\begin{abstract}
\begin{quote}
Many modern applications that deal with sensitive data, such as healthcare and government services, outsource computation to cloud platforms. In such untrusted environments, privacy is of vital importance. One solution to this problem is homomorphic encryption (HE), a family of cryptographic schemes that support certain algebraic operations on encrypted data without the need for decryption. However, despite major advancements, encryption in modern HE schemes still comes with a non-trivial computational overhead that can hamper data-intensive workloads. To resolve this, recent research has shown that leveraging caching techniques, such as Rache, can significantly enhance the performance of HE schemes while maintaining security. Rache unfortunately displays a key limitation in the time complexity of its caching procedure, which scales with the size of the plaintext space. Smuche is another caching scheme that simultaneously improves the scalability of the caching procedure and turns the encryption process into a constant-time operation, utilizing only a single scalar multiplication. Even still, more can be done. In this paper, we present an encryption method we call ``Zinc" which entirely forgoes the multiple caching process, replacing it with a single scalar addition, and then injecting randomness that takes constant time with respect to the plaintext space. This injection of randomness is similar to Smuche, and a great improvement from Rache, allowing Zinc to achieve efficiency without compromising security. We implement the scheme using Microsoft SEAL and compare its performance to vanilla CKKS. The results show further progress toward a reasonable level of practicality for homomorphic encryption.
\end{quote}
\end{abstract}

\section{Introduction}

\subsection{Background and Motivation}

Systems handling sensitive numeric data and employed in untrustworthy environments must ensure that said data is securely stored without incurring an unacceptable performance overhead. In particular, this issue has become evident in machine learning, where some models trained on distributed platforms are expected to maintain privacy of the data being trained on. Even sharing the models themselves is insecure; it has been shown that one can acquire the training data only from model gradients \cite{NEURIPS2019_60a6c400}. 

The advent of fully homomorphic encryption (FHE) has provided a plausible solution to this problem. Models can be both trained and shared in fully encrypted form, providing strong security guarantees in public settings. Nonetheless, the heavy computational cost remains an issue, driving a recent push to design procedures that mitigate this cost as much as possible, as in Rache \cite{10.1145/3588920}.

Rache, however, comes with the caveat that caching too many ciphertexts can result in overall \textit{reduction}, pointing to scalability issues. Another issue is that Rache only works on integer values, with it being difficult to extend the scheme to accommodate for real values. Hence we turn our attention to the CKKS \cite{asiacrypt-2017-28280} scheme, which does support encryption of real (and complex) values. When it comes to implementations of vanilla CKKS, libraries such as Microsoft SEAL \cite{sealcrypto} are already heavily optimized to the point that it is difficult to improve much further. We propose an approach that allows for at least factor-of-two reduction in encryption times, with a negligible pre-computation overhead, and without sacrificing strong security guarantees.

\subsection{Contributions}

Our proposed scheme exploits the structure of modern homomorphic encryption scheme ciphertexts in order to design an optimized version that, while similar in essence to prior caching schemes, looks completely different beneath the surface. Our contributions are threefold:
\begin{enumerate}
    \item We propose an algorithm, dubbed Zinc, which is able to produce an encrypted HE ciphertext with both constant-time caching and constant-time randomization. Zinc is the first of cached HE algorithms which is able to fully operate in constant time with respect to the ciphertext space.
    \item We formally define our security goals and threat model, and then state the correctness and security of the Zinc scheme, which is proven in a typical reduction format.
    \item We implement Zinc using Microsoft SEAL \cite{sealcrypto} and perform empirical analysis showing that our scheme achieves at least factor-of-two reductions in encryption time against both CKKS \cite{asiacrypt-2017-28280} and Rache \cite{10.1145/3588920} as baselines.
\end{enumerate}

\section{Preliminaries and Related Work}

\subsection{Homomorphic Encryption}

\textit{Homomorphic} is a name that comes from the idea of a \textit{homomorphism}, which itself comes from the theory of groups. A group is defined as follows.

\begin{definition}
    A \textbf{group} is a tuple consisting of a non-empty set $G$ and a binary operation $\circ$, denoted $(G, \circ)$ that meets four axioms:
    \begin{enumerate}
        \item (Closure) If $a, b \in G$, then $a \circ b \in G$.
        \item (Associativity) If $a, b, c \in G$, then $(a \circ b) \circ c = a \circ (b \circ c)$.
        \item (Identity) There exists $e \in G$ such that $e \circ a = a \circ e = a$ for all $a \in G$.
        \item (Inverse) For every $a \in G$, there exists $i \in G$ such that $a \circ i = i \circ a = e$.
    \end{enumerate}
\end{definition}

This sets us up to define functions between groups. In our case, we would like for our function to have a ``structure-preservation" or ``operation-preservation" property. For this, we define the homomorphism:

\begin{definition}
    Let $(F, \oplus)$ and $(G, \otimes)$ be groups. A function $\phi : F \to G$ is called a \textbf{homomorphism} just in case
    \begin{equation*}
        \phi(f_1 \oplus f_2) = \phi(f_1) \otimes \phi(f_2) 
    \end{equation*}
    for all $f_1, f_2 \in F$.
\end{definition}

Homomorphisms show up everywhere in mathematics. For demonstration, we provide a very simple example.

\begin{example}
    Consider the groups $F = (\mathbb{R}, +)$ and $G = (\mathbb{R}^+, \times)$, where $\mathbb{R}$ and $\mathbb{R}^+$ denote the reals and the positive reals, respectively. Define $\phi(x) = 2^x$. It is clear that 
    \begin{equation*}
        \phi(a + b) = 2^{a + b} = 2^a \times 2^b = \phi(a) \times \phi(b)
    \end{equation*}
    so $\phi$ is a homomorphism.
\end{example}

Homomorphic encryption makes use of the properties of homomorphism between a \textit{plaintext space} and the \textit{ciphertext space}. More precisely, we would like for a the encryption function to be a homomorphism, i.e. 

\begin{equation}\label{eq:1}
    \mathsf{Enc}(a + b) = \mathsf{Enc}(a) \oplus \mathsf{Enc}(b)
\end{equation}
for plaintexts $a, b$, an encryption function $\mathsf{Enc}$, and an operation $\oplus$ over the ciphertexts. However, we note that in reality, Eq. \eqref{eq:1} can only be satisfied in the case that the encryption scheme is deterministic. More realistically, we expect that an HE scheme will have randomized ciphertexts, leading to the condition

\begin{equation}\label{eq:2}
    \mathsf{Dec}(\mathsf{Enc}(a) \oplus \mathsf{Enc}(b)) = a + b
\end{equation}
for the decryption function $\mathsf{Dec}$. Such a scheme is called \textit{additively} homomorphic. An example of such a scheme is Paillier \cite{10.1007/3-540-48910-X_16}, which is one of the schemes that Rache was implemented over. A \textit{multiplicatively} homomorphic scheme is defined analogously:

\begin{equation}\label{eq:3}
    \mathsf{Dec}(\mathsf{Enc}(a) \otimes \mathsf{Enc}(b)) = a \times b
\end{equation}
for another operation $\otimes$ over the ciphertexts.

A scheme which exhibits only one or the other property, but not both, is called \textit{partially homomorphic}. A scheme which exhibits both simultaneously is called \textit{fully homomorphic}. Constructing fully homomorphic encryption (FHE) schemes proved to be exceptionally difficult until Gentry \cite{10.1145/1536414.1536440} showed that it was possible using ideal lattices.

The most popular of modern FHE schemes include BGV \cite{cryptoeprint:2011/277}, BFV \cite{cryptoeprint:2012/144}, and CKKS \cite{asiacrypt-2017-28280}. These schemes are based on hardness assumptions of the learning with errors (LWE) \cite{10.1145/1568318.1568324} or the ring variant ring-LWE (RWLE) \cite{10.1007/978-3-642-13190-5_1} problems. This generation of FHE makes sizable leaps in performance, but is not quite as good as we would like it to be just yet.

\subsection{Provable Security}

It is necessary to define our security goals so that we may demonstrate the security of our scheme. Informally, we suppose that an adversary with (reasonably) limited computational power attempts to ``break" the scheme. The adversary wins if there is a non-negligibly high chance of the adversary guessing a plaintext correctly for a given ciphertext. This notion is formalized by the \textit{chosen-plaintext attack} (CPA).

\begin{definition}(Chosen-Plaintext Attack)
    For a given security parameter $\kappa$, such as the bitstring length of a key, an adversary $\mathcal{A}$ may obtain $poly(\kappa)$ plaintext-ciphertext pairs, denoted $(m, c)$, where $m$ is picked by $\mathcal{A}$ arbitrarily, and $poly(\cdot)$ is a polynomial function. $\mathcal{A}$ then tries to decrypt some $c'$ that is not included in the ciphertexts that were already computed.
\end{definition}

The restriction of polynomially-bounded resources is in line with the usual idea of efficiency; the more formal way to express this is that $\mathcal{A}$ is a probabilistic polynomial-time (PPT) Turing machine. Security under this kind of attack is coined \textit{indistinguishability} under CPA, or IND-CPA.

\begin{definition}(IND-CPA)\label{ind-cpa}
    Suppose an adversary $\mathcal{A}$ begins a chosen-plaintext attack against a challenger $\mathcal{C}$. After acquiring $poly(\kappa)$ ciphertext-plaintext pairs, $\mathcal{A}$ sends two plaintexts, $m_1$ and $m_2$, to $\mathcal{C}$. $\mathcal{C}$ flips a coin $b \in \{ 0,1 \}$, and sends back $\mathsf{Enc}(m_b)$. A scheme is called \textbf{IND-CPA} if $\mathcal{A}$ cannot guess which plaintext was encrypted with significantly better than 1/2 probability.
\end{definition}

The degree of ``significance" is quantified formally like so:

\begin{definition}
    A function $\mu(\cdot)$ is considered \textbf{negligible} if, for any polynomial $poly(n)$, 
    \begin{equation*}
        \mu(n) < \frac{1}{poly(n)}
    \end{equation*}
    for sufficiently large $n$.
\end{definition}

So a scheme is IND-CPA if the adversary cannot guess with better than $1/2 + \mu(\kappa)$ probability, with $\mu$ a negligible function, and $\kappa$ a security parameter. We will continue to use $\kappa$ to denote a security parameter throughout the paper. The following lemma about negligible functions will also be used in a later section.

\begin{lemma}\label{lem:1}
    Let $\mu(\cdot)$ be a negligible function and $\gamma(\cdot)$ a non-negligible function. Then $\mu(n) \leq \gamma(n)$ for sufficiently large $n$.
\end{lemma}

\begin{proof}
    By definition,
    \begin{equation*}
        \mu(n) < \frac{1}{poly(n)} \quad \text{and} \quad \gamma(n) \geq \frac{1}{poly(n)}
    \end{equation*}
    So, 
    \begin{equation*}
        \mu(n) < \frac{1}{poly(n)} \leq \gamma(n)
    \end{equation*}
    Giving the claim.
\end{proof}

\subsection{Polynomials and Distributions}

We denote by $\mathbb{Z}_q/(X^N + 1)$ the polynomial ring whose elements are polynomials with coefficients modulo $q$ (called the \textit{coefficient modulus}) and of degree $N$. A plaintext consists of one ring element (polynomial) and a ciphertext is a pair of polynomials. Then we define three key distributions:
\begin{enumerate}
    \item $R_q$ is the uniform distribution over $\mathbb{Z}_q/(X^N + 1)$.
    \item $\chi$ is a discrete Gaussian distribution with $\sigma = 3.2$.
    \item $R_2$ is used to sample polynomials with coefficients in $\{-1, 0, 1\}$.
\end{enumerate}
We will use the notation $[\cdot]_q$ to denote that polynomial arithmetic should all be done modulo $q$. Additionally, $\lVert v \rVert_\infty$, the infinity norm of a polynomial $v$, denotes the largest coefficient of $v$.

\subsection{Rache: Radix-Additive Cached HE} 

Our work naturally builds off ongoing inquiry into methods for optimizing homomorphic encryption without giving up on security guarantees. The current state-of-the-art is Rache \cite{10.1145/3588920}. The key observation made by Rache is that homomorphic additions are much less computationally costly than encrypting fresh ciphertexts is. This is exploited in three steps, which are roughly explained. We refer readers to the original paper for more detail.

\subsubsection{Precomputation} 

Rache selects a radix $r$, and computes the $r$-power series ciphertexts, caching $\mathsf{Enc}(r^i)$, called \textit{pivots}, for $r^i \leq 2^\kappa$. This typically will encompass the entire plaintext space. The simplest choice of $r$ is $2$, though the paper uses a heuristically-chosen value of $r = 6$.

\subsubsection{Construction} 

The ciphertext is constructed by first representing a plaintext $m$ in its base-$r$ form, which may be represented as the array $digits$ of the digits of $m$ in base-$r$. The pivots are then summed appropriately.

\subsubsection{Randomization}

Before this stage, the constructed ciphertext is deterministic, so we must add some randomness, and this is done by constructing a random $\mathsf{Enc}(0)$. In Rache, this is done by iterating over every $r^i$, and randomly choosing whether or not to add $\mathsf{Enc}(r^i)$. If it is added, then $\mathsf{Enc}(r^i)$ is subtracted $r$ times.

\section{Zinc}

Zinc is founded on the same ideas as Rache, but largely gets rid of extensive computations at all three stages. In particular, we notice that in CKKS, plaintext-ciphertext operations take even less time than ciphertext-ciphertext operations. There are still three steps in Zinc, but they are far more straightforward than in Rache.

Usually, encrypting a fresh ciphertext takes two polynomial multiplications and at least two polynomial additions. Our crucial observation is that we can remove the need for both of the multiplication operations by maintaining a pre-encrypted ciphertext, and adding new randomness in the form of two polynomial additions.

\subsection{Scheme Description}

\subsubsection{Precomputation}

Zinc precomputes and caches a single $\mathsf{Enc}(0)$. This is all that needs to occur at this stage.

\subsubsection{Construction}

We assume that the message $m$ has already been encoded into a suitable plaintext format, and we will reuse the $+$ symbol for plaintext-ciphertext multiplication. It is clear that $0 + m = m$, so utilizing homomorphism, we may compute
\begin{equation}
    \mathsf{Enc}(m) = \mathsf{Enc}(0) + m
\end{equation}

One should note that this plaintext-ciphertext addition step corresponds precisely to the structure of a ciphertext, as is shown in a later section (see Eq. \eqref{eq:8}). That being said, since $m$ is a deterministic encoding and $\mathsf{Enc}(0)$ is cached, this results in a deterministic ciphertext for $\mathsf{Enc}(m)$.

\subsubsection{Randomization}

At this point, we are no longer able to treat the underlying HE scheme as a black box. Hence, we will cover the ciphertext structure of modern HE schemes, such as the aforementioned BFV \cite{cryptoeprint:2012/144} and CKKS \cite{asiacrypt-2017-28280}, at a high-level in order to understand how the randomization step works.

A ciphertext $c$ in FHE schemes like those above consists of a tuple of polynomials, 
\begin{equation}
    c = (c_1, c_2)
\end{equation}
such that $c_1, c_2 \in \mathbb{Z}_q / (X^N + 1)$. Most of the message is contained in $c_1$, as we will see momentarily. The secret key $sk$ is just a small ``ternary" polynomial. That is, $sk \leftarrow R_2$, where $R_2$ is a distribution from which we sample polynomials with uniformly distributed coefficients in $\{-1, 0, 1\}$.

The public key $pk$ has the same general structure as a ciphertext, 
\begin{equation}
    pk = (pk_1, pk_2)
\end{equation}
where
\begin{equation}\label{eq:7}
\begin{cases}
    pk_1 = [-a \cdot sk + e]_q \\
    pk_2 = a \leftarrow R_q
\end{cases}
\end{equation}
given $e \leftarrow \chi$. We then sample $u \leftarrow R_2$ and $e_1, e_2 \leftarrow \chi$, and use $pk$ to define $c_1$ and $c_2$ for a message $m$:
\begin{equation}\label{eq:8}
\begin{cases}
    c_1 = [pk_1 \cdot u + e_1 + m]_q \\
    c_2 = [pk_1 \cdot u + e_2]_q
\end{cases}
\end{equation}

This finally leads us to the critical observation. Note that both of $e_j$ are simply small polynomials$-$we can simply add another polynomial to each component of the ciphertext to randomize it. Specifically, we sample $e_1'$ and $e_2'$, returning the randomized ciphertext
\begin{equation}
    \mathsf{Enc}_{z}(m) = c' = ([c_1 + e_1']_q, [c_2 + e_2']_q)
\end{equation}
Where $c = (c_1, c_2)$ is the ciphertext computed in the construction step, and $\mathsf{Enc}_z$ the Zinc encryption function. This concludes the randomization step.

\subsection{Correctness}
The entire encryption procedure is straightforwardly formalized in Algorithm \ref{alg:encz}.
\begin{algorithm}
\caption{$\mathsf{Enc}_z$: Zinc Encryption}\label{alg:encz}
\begin{algorithmic}[1]
\item[\textbf{Input:}] A plaintext $m$ and a cached $\mathsf{Enc}(0)$, $zero$
\item[\textbf{Output:}] A ciphertext $c'$ satisfying $\mathsf{Dec}(c') = m$
\State $c = zero + m$
\State $e_1', e_2' \gets \chi$
\State $c' =  ([c_1 + e_1']_q, [c_2 + e_2']_q)$
\end{algorithmic}
\end{algorithm}
Before we continue to its correctness, decryption is defined as
\begin{equation}\label{eq:dec}
    \mathsf{Dec}(c) = [c_1 + c_2 \cdot sk]_q
\end{equation}
We show that the $\mathsf{Enc}_z(m)$ is a valid encryption for a plaintext $m$.

\begin{theorem}\label{thm:1}
    Let $q$ be the coefficient modulus, and $m$ a plaintext. Then $\mathsf{Dec}(\mathsf{Enc}_z(m)) \equiv_q m$.
\end{theorem}

\begin{proof}
    The proof proceeds by algebra.
    \begin{align*}
        \mathsf{Dec}&(\mathsf{Enc}_z(m)) \\ 
            &\equiv_q c_1 + c_2 \cdot sk \\
            &\equiv_q pk_1 \cdot u + e_1  + e_1' + m + (pk_2 \cdot u + e_2 + e_2') \cdot sk
    \end{align*}
    For simplicity, we write $e_a = e_1 + e_1'$ and $e_b = e_2 + e_2'$. Then we expand $pk_1$ and $pk_2$ with Eq. \eqref{eq:7}:
    \begin{align*}
        pk_1 \cdot u &+ e_1  + e_1' + m + (pk_2 \cdot u + e_2 + e_2') \cdot sk \\
            &\equiv_q pk_1 \cdot u + e_a + m + (pk_2 \cdot u + e_b) \cdot sk \\
            &\equiv_q (-a \cdot sk + e) \cdot u + e_a + m + (a \cdot u + e_b) \cdot sk \\
            &\equiv_q e \cdot u + m + e_a + e_b \cdot sk \\
            &\equiv_q m + v
    \end{align*}
    Where $\lVert v \rVert_\infty$ is very small, as it is a sum of small polynomials. Thus, this is a correct decryption of $m$.
\end{proof}

\subsection{Security}

We state a homomorphic encryption scheme as the quintuple
\begin{equation*}
    \Pi = (\mathsf{Gen}, \mathsf{Enc}, \mathsf{Dec}, \oplus, \otimes)
\end{equation*}
Where $\mathsf{Gen}$ is a function generating an encryption key $k$ of length $\kappa$, $\mathsf{Enc}$ and $\mathsf{Dec}$ are encryption and decryption functions parameterized by $k$. $\oplus$ and $\otimes$ denote the additive and multiplicative homomorphic operations, respectively. We then define its Zinc-extension as
\begin{equation}
    \widetilde{\Pi} = (\mathsf{Gen}, \mathsf{Enc}_z, \mathsf{Dec})
\end{equation}
Where $\mathsf{Enc}_z$ is defined as in Algorithm \ref{alg:encz}. We show the security of the scheme in what follows.

\begin{theorem}
    If the underlying HE scheme $\Pi$ is IND-CPA, then its Zinc-extension $\widetilde{\Pi}$ is IND-CPA.
\end{theorem}

\begin{proof}
    The problem of breaking $\widetilde \Pi$ is reduced to that of breaking $\Pi$. Assuming $\Pi$ is IND-CPA, suppose for contradiction that there exists a PPT adversary $\mathcal{A}$ for $\widetilde{\Pi}$ that can distinguish ciphertexts with non-negligible probability.
    
    Without loss of generality, assume $\mathcal{A}$ would choose the plaintexts $M_1$ and $M_2$. We construct an adversary $\mathcal{B}$ for $\Pi$ that runs the following protocol:
    \begin{enumerate}
        \item $\mathcal{B}$ sends the plaintexts $M_1$ and $M_2$ to the challenger $\mathcal{C}$.
        \item $\mathcal{C}$ sends back a ciphertext $\mathsf{Enc}(M_b) = (c_1, c_2)$ for a random $b \in \{0,1\}$.
        \item $\mathcal{B}$ generates two polynomials, $e_1', e_2' \leftarrow \chi$. Then, it constructs the ciphertext $c' = ([c_1 + e_1']_q, [c_2 + e_2']_q)$.
        \item $\mathcal{B}$ asks $\mathcal{A}$ to distinguish which of $M_0$ or $M_1$ is encrypted by $c'$.
        \item $\mathcal{A}$ makes a guess $b'$, and $\mathcal{B}$ outputs $b'$.
    \end{enumerate}
    Generating $e_1'$ and $e_2'$ are clearly polynomial time with respect to the plaintext space: given a polynomial modulus degree of $N$, this can be done $\mathcal{O}(N)$ time. Since $\mathcal{B}$ then simply asks $\mathcal{A}$ to perform its own operations, $\mathcal{B}$ is PPT.
    
    Now, notice that the ciphertext $c'$ is in exactly the form that is encrypted by $\mathsf{Enc}_z(M_b)$, and by Theorem \ref{thm:1}, it is a valid encryption of $M_b$, so $\mathcal{A}$ may guess on it. Let $CPA_X^\mathcal{X}$ denote the CPA indistinguishability experiment as in Definition \ref{ind-cpa} with scheme $X$ and adversary $\mathcal{X}$. By supposition, $\mathcal{A}$ is able to guess correctly with probability
    \begin{equation*}
        \Pr[CPA_{\widetilde\Pi}^\mathcal{A} = 1] \geq \frac{1}{2} + \gamma(\kappa)
    \end{equation*}
    for a non-negligible $\gamma$. Since $\mathcal{B}$ guesses based on what $\mathcal{A}$ does, 
    \begin{equation*}
        \Pr[CPA_\Pi^\mathcal{B} = 1] = \Pr[CPA_{\widetilde\Pi}^\mathcal{A} = 1] \geq \frac{1}{2} + \gamma(\kappa)
    \end{equation*}
    However, it is assumed that $\Pi$ is IND-CPA. So
    \begin{equation*}
        \Pr[CPA_\Pi^\mathcal{B} = 1] < \frac{1}{2} + \mu(\kappa)
    \end{equation*}
    for a negligible $\mu$. This implies that
    \begin{equation*}
        \frac{1}{2} + \mu(\kappa) > \frac{1}{2} + \gamma(\kappa) \implies \mu(\kappa) > \gamma(\kappa)
    \end{equation*}
    which cannot happen by Lemma \ref{lem:1}. We conclude that the scheme $\widetilde\Pi$ is IND-CPA.
\end{proof}

In practice, most existing schemes are already proven to be IND-CPA. Zinc is built on CKKS \cite{asiacrypt-2017-28280}, which itself is known to be IND-CPA.

\section{Evaluation}

\subsection{Theoretical Analysis}

Many modern implementations of the CKKS scheme use a modified, RNS accelerated \cite{10.1007/978-3-030-10970-7_16} version, which performs a discrete Fourier transform (DFT) on the polynomials. This comes with the advantage that polynomial multiplication becomes faster. But does become a bottleneck in the Zinc scheme. 

To see why, recall that vanilla CKKS performs two polynomial multiplications and three additions. The additions can be done in a linear manner, so it only remains that
\begin{enumerate}
    \item A DFT has to be performed on \textit{each} polynomial, taking $\mathcal{O}(N \log N)$ time.
    \item After the DFT, we perform polynomial multiplication, taking $\mathcal{O}(N \log N)$ time.
\end{enumerate}
So encryption in CKKS is (roughly) $\mathcal{O}(N \log N)$.

To make Zinc compatible with RNS accelerated operations, we must also perform a DFT on both $e_1'$ and $e_2'$, making encryption in Zinc theoretically have the same quasilinear complexity that CKKS does. 

We remark that, in practice, Zinc reduces the number of DFTs performed to only two, and has no need for multiplication, resulting in at least a factor-of-two performance increase regardless of the asymptotic complexity remaining the same.

\subsection{Scheme Flexibility}

For clarity, we emphasize that Zinc can be implemented on top of any HE scheme with similar characteristics to CKKS, including BFV \cite{cryptoeprint:2012/144} and BGV \cite{cryptoeprint:2011/277}, both of which are also implemented in Microsoft SEAL. For demonstration, we will quickly look at the ciphertext structures of BFV and BGV.

In both BFV and BGV, additional consideration is given to the plaintext space, which is similar to the ciphertext space, defined as $\mathbb{Z}_t/(X^N + 1)$ for a coefficient modulus $t \ll q$ (recall that $q$ is the ciphertext space coefficient modulus). 

In BFV, we allow the error to grow in the \textit{lower} bits of the ciphertext, so we scale up the message by a factor $\Delta = \lfloor q/t \rfloor$. The ciphertext is then
\begin{equation}\label{eq:12}
\begin{cases}
    c_1 = [pk_1 \cdot u + e_1 + \Delta \cdot m]_q \\
    c_2 = [pk_1 \cdot u + e_2]_q
\end{cases}
\end{equation}
So we may use the Zinc without modification in BFV.

In BGV, the opposite is done, though the idea remains the same. We now allow the error to grow in the \textit{upper} bits. Thus, we now scale the error rather than the message. Formally, the ciphertext becomes
\begin{equation}\label{eq:13}
\begin{cases}
    c_1 = [pk_1 \cdot u + t \cdot e_1 + m]_q \\
    c_2 = [pk_1 \cdot u + t \cdot e_2]_q
\end{cases}
\end{equation}
Only a minor alteration needs to be made for Zinc to be compatible with BGV. That is, in order to account for the scaled error, we add $t \cdot e_1'$ and $t \cdot e_2'$.

Correctness can be easily verified by the same algebraic process to Theorem $\ref{thm:1}$, albeit with a different decryption algorithm. The details of decryption in each scheme can be found in their original papers, as well as plaintext encoding and decoding.

\subsection{System Implementation}

The Zinc scheme is implemented on top of CKKS \cite{asiacrypt-2017-28280} in the Microsoft SEAL library \cite{sealcrypto}. We additionally implement a baseline form of Rache \cite{10.1145/3588920} for comparison, showing that performance benefits can be gained over the current state-of-the-art. Zinc is written in the C++ language, compiled with GCC. 

Microsoft SEAL provides several functions that allow Zinc to be implemented with ease. Parameters are fixed between each scheme to ensure fair testing. Specifically, we use a polynomial modulus degree of $N = 32768$, the largest $N$ defined by the homomorphic encryption standard (HES) \cite{cryptoeprint:2019/939}. 

The coefficient modulus is determined automatically using SEAL's \texttt{CoeffModulus::BFVDefault(N)} method. In general, these are prime numbers as large as $2^{64} - 1$, and may also be chosen manually. The HES defines the maximum value of $\log_2 q$ to be 881. In SEAL, the default chosen is a list of primes around $2^{55}$ in size, with the last being close to $2^{56}$ to reach 881 bits total. CKKS uses a plaintext scale (similar to BFV's scale) that is chosen before encryption. SEAL recommends that the scale, $\Delta$, is chosen to be close to the intermediate primes. For our implementation, that means $\Delta = 2^{55}$. 

Our implementation of Rache in SEAL is the batch version, using a radix of $r = 2$. 

\subsection{Experimental Setup}

Microsoft SEAL is not inherently multi-threaded, so for the purpose of these tests, no parallel processing features are used in the implementation. The benchmarks were run on a single node on CloudLab \cite{Duplyakin+:ATC19} with the following specifications: 128 GB of physical memory (RAM), a single AMD EPYC 7302P 16-Core Processor @ 3GHz, on Ubuntu 22.04 LTS. 

We perform a set of micro benchmarks on random numbers testing the performance of an implementation of Rache \cite{10.1145/3588920} in SEAL, displaying the scalability limitations of Rache for a modern FHE scheme, particularly in heavily optimized libraries like SEAL. Additionally, we compare their performance on three real-world datasets to simulate real-world application:
\begin{enumerate}
    \item U.S. Covid-19 Statistics \cite{covid19}, containing 341 days of 16 metrics, such as \textit{death increase}, \textit{positive increase}, and \textit{hospitalized increase}.
    \item A sample of Bitcoin trade volume \cite{bitcoin} from February 2013 to January 2022, totalling 1,086 floating-point numbers taken on a 3-day basis within the time period.
    \item Human genome reference number 38 \cite{hg38}, also called \textit{hg38}, consisting in 34,424 rows of singular attributes, including \textit{transcription positions}, \textit{coding regions}, and \textit{number of exons}, last updated in March 2020.
\end{enumerate}
All singular benchmarks reported were done five times each, and the average of the five runs is reported.

\begin{table}
\centering
\begin{tabular}{r r r}
 \hline
    $n$ Pivots & Rache $\mathsf{RHE}$ (ms) & Ratio over CKKS \\ 
 \hline
     4 &   8.03 & 0.14 \\ 
     8 &  15.51 & 0.27 \\
    16 &  29.85 & 0.53 \\
    32 &  58.35 & 1.01 \\
    64 & 109.81 & 1.91 \\
 \hline
\end{tabular}
 \caption{Scalability of Rache in SEAL (CKKS $\sim 57.6$ ms)}
 \label{table:1}
\end{table}

\subsubsection{Scalability of Rache}

Table \ref{table:1} shows a set of micro benchmarks, encrypting 1,024 integers, that display the scalability issues of Rache \cite{10.1145/3588920}. We vary the number of \textit{pivots}, i.e. the number of precomputed and stored ciphertexts, note the performance of Rache's encryption, and report the ratio of this value compared to CKKS \cite{asiacrypt-2017-28280} for the same numbers. On average, vanilla CKKS encryption took around 57.6 ms.

In the batch version of Rache, the implementer is expected to cache a number of ciphertexts up to $2^\kappa - 1$. Recall that the security parameter $\kappa$ is related to the number of bits in a key. The HES \cite{cryptoeprint:2019/939} defines this at three different levels: 128-bit, 192-bit, and 256-bit. Thus, even if Rache were to outperform CKKS \textit{with} multiprocessing at the parameters labeled above, this would not be the case for a large enough security parameter. Since we use a radix of $r = 2$, this would be
\begin{equation*}
    \lfloor \log_2(2^\kappa - 1) \rfloor \approx \kappa - 1
\end{equation*}
For a large enough security parameter, and since $\kappa$ is an integer. In the case of the HES, this would require us to cache 127, 191, and 255 ciphertexts in total in order for proper security to be maintained.

\subsubsection{Real-World Data Sets} 

The total performance over an entire data set for all three schemes (CKKS, Rache, and Zinc) is shown in Table \ref{table:2}. By default, $n = 32$ pivots are used on all three data sets due to the largest values in the Bitcoin dataset \cite{bitcoin} being very large. A secondary test is done with $n = 9$, as this is the absolute minimum number of pivots that can accommodate the values in the hg38 \cite{hg38} dataset (less than $2^9$), making the comparison as fair as possible.

When Rache was given an absolute minimum of $n = 9$ pivots, it was able to outperform even Zinc, exhibiting a total encryption time reduction of about 39.3\%. Recall, however, that the choice of number of cached pivots in the batch variant of Rache roughly corresponds to a security parameter. We stress that with a value as small as $\kappa = 9$, the security of the scheme in this state is dubious at best, especially given that, by default, Microsoft SEAL \cite{sealcrypto} uses a security parameter of $128$.

In every other case, we see that Rache only marginally outperforms CKKS on all three datasets with $n = 32$ pivots. Additionally, on all three datasets, Zinc defeats both Rache and CKKS by at least a reduction of 50\% in total encryption time, which is consistent with the theoretical analysis in a prior section.

\begin{table}
\centering
\begin{tabular}{l r r r}
 \hline
          &  Covid-19 & Bitcoin & Human Genome 38 \\ 
 \hline
    CKKS  & 19.2 sec & 62.2 sec & 1983.6 sec \\
    Rache & 18.4 sec & 61.2 sec & 1789.4 sec \\
    $\hookrightarrow n = 9$ & N/A & N/A & 493.2 sec \\
    Zinc  &  8.2 sec & 24.8 sec &  812.6 sec \\
 \hline
\end{tabular}
 \caption{Total Encryption Performance on Real-World Data}
 \label{table:2}
\end{table}

\section{Conclusion}

We have introduced Zinc, a constant-time, extremely low precomputation cost encryption scheme that both theoretically and empirically has shown at least factor-of-two performance benefits over an optimized implementation of a modern fully homomorphic encryption scheme, as well as the current state-of-the-art in optimization methods for homomorphic encryption. We addressed the correctness of Zinc, and showed that Zinc meets the IND-CPA security goal as long as the underlying scheme is also secure. Zinc also marks a large reduction in system complexity compared to the current state-of-the-art.

It is conceivable that further optimizations can be made to the Zinc scheme. For instance, it was discussed that, despite the large boost in performance, the overall runtime complexity of Zinc's encryption is still $\mathcal{O}(N \log N)$ due to the need to perform a discrete Fourier transform. Generating randomized error polynomials directly in the Fourier space would allow us to bypass this step, reducing the overall complexity to $\mathcal{O}(N)$. Additionally, applying multiprocessing, particularly with the polynomial addition step, would undoubtedly be another avenue for increasing performance. 

\bibliographystyle{aaai} \bibliography{references}

\end{document}